\documentclass[letterpaper, 10 pt, conference]{ieeeconf}  

\IEEEoverridecommandlockouts                              

\usepackage{graphics} 
\usepackage{epsfig} 
\usepackage{times} 
\usepackage{amsmath} 
\usepackage{amssymb}  

\usepackage{mathtools}
\usepackage{mathrsfs} 

\usepackage{cite}

\newtheorem{theorem}{Theorem}[section]

\newtheorem{proposition}[theorem]{Proposition}
\newtheorem{lemma}[theorem]{Lemma}

\newtheorem{problem}{Problem}[section]

\newtheorem{definition}{Definition}

\newcommand{\PoA}{{\rm{PoA}}}
\newcommand{\NE}{{\rm{NE}}}
\newcommand{\Ne}{{\rm{ne}}}
\newcommand{\WNE}{{\rm{WNE}}}

\newcommand{\opt}{{\rm{opt}}}
\newcommand{\aaa}{{\mathcal{A}}}

\newcommand{\leftparen}{\left(}
\newcommand{\leftcurly}{\left\{}
\newcommand{\leftbracket}{\left[}

\newcommand{\rightparen}{\right)}
\newcommand{\rightcurly}{\right\}}
\newcommand{\rightbracket}{\right]}

\title{\LARGE \bf On the Fragility of Worst-Case Nash Equilibria in Atomic Congestion Games
}

\author{Colton Hill, Brandon Collins, and Philip N. Brown
\thanks{*Research was sponsored by the Air Force Office of Scientific Research under award number FA9550-23-1-0171, by the National Science Foundation under grant number ECCS-2440836, and by NASA under grant number 80NSSC25M7102.}
\thanks{C. Hill and P. N. Brown are with the University of Colorado at Colorado Springs, CO 80918, USA.
{\tt\small \{chill13, pbrown2\}@uccs.edu.}
P. N. Brown completed this work while a Visiting Professor at Politecnico di Torino, Turin, Italy.
B. Collins is with the Thayer School of Engineering, Dartmouth College, Hanover, NH 03755, USA.
{\tt\small Brandon.C.Collins@dartmouth.edu.}}%
}

\begin{document}

\maketitle
\thispagestyle{empty}
\pagestyle{empty}

\begin{abstract}

Aggregate performance in smart mobility systems depends heavily on the emergent behavior of selfish, resource-sharing agents that participate within the systems.
As a result, recent work has focused on how a system designer can leverage incentives to influence behavior so that system cost (e.g., traffic congestion) is minimized.
These results show that worst-case equilibria can be quite inefficient compared to system-optimal allocations.
However, it is unclear to what extent agents are ``satisfied'' with their decisions in these worst-case scenarios.
We demonstrate that in any incentivized atomic congestion game, agents' aggregate satisfaction at equilibrium (relative to their actions in an optimal allocation) is correlated with the efficiency of the corresponding system cost, in the sense that if agents are very satisfied with their equilibrium choices, the equilibrium must be relatively efficient.
Further, we show that worst-case Nash equilibria are fragile, as every agent is indifferent between their action in a worst-case equilibrium and their action in a system-optimal allocation.
In summary, at equilibrium, either agents are highly satisfied with their decisions or their decisions are highly inefficient, but both cannot be true simultaneously.
This work adds to recent results for other classes of games which indicate that worst-case equilibrium efficiency guarantees only occur when agents are indifferent about their decisions.

\end{abstract}

\section{Introduction}

Modern social systems increasingly depend on a shared infrastructure among agents, such as transportation and communication networks, and electric power grids~\cite{rinaldi2001identifying}.
The aggregate performance of these systems is often shaped by the decisions of self-interested agents whose behavior can harm system efficiency~\cite{roughgardenSelfishRoutingPrice2005}.
An approach for quantifying system inefficiency due to selfish agents that has received attention in recent literature is the \emph{price of anarchy} (PoA)~\cite{pigou2017economics, christodoulou2005price, paccagnan2021optimal}.
The PoA is defined as the worst-case ratio between the system cost induced by agent behavior at equilibrium and the system cost when agents are allocated optimally.
Bounds on the PoA when agents are uninfluenced are well-known, and provide a system designer with efficiency guarantees for a congestible network in the presence of selfish agents~\cite{roughgarden2002bad, awerbuch2005price, perakis2007price, hoeksma2011price, roughgardenIntrinsicRobustnessPrice2015}.

Centralized control can reduce the PoA, but is often infeasible, so a key challenge is whether a system designer can design incentives in a decentralized manner that provide aggregate efficiency guarantees.
This has been studied broadly in recent literature on mechanism design and learning in congestion games, including parameter estimation~\cite{panageas2023semi}, model learning in resource-allocation games~\cite{dong2023taming}, and incentive mechanisms for improving the PoA~\cite{fotakis2010stackelberg}.
Further, tolling mechanisms have become a common tool for influencing agent behavior and reducing network inefficiency~\cite{bilo2019dynamic, ferguson2021effectiveness}.

The design of tolling mechanisms focuses on taxing resources that are sensitive to the number of agents utilizing the resource (e.g., by applying a toll to a congestible roadway).
Tolls have been studied in a variety of settings, including adaptive toll design and robustness to misspecified tolls in congestion games~\cite{chiu2024parameter, chiu2025robustness}.
Given sufficient knowledge of the system and agents' preferences, a system designer can design optimal tolls within a class of mechanisms such that no other toll in the class provides greater improvement in the PoA~\cite{bilo2019dynamic, caragiannis2010taxes}.
Moreover, \emph{local} tolls, which are determined only by resource-specific information, provide performance guarantees similar to those of globally optimal tolls~\cite{paccagnan2021optimal}.

Optimal tolling mechanisms improve the PoA relative to when no tolls are applied, but they do not eliminate inefficiency altogether, as equilibrium outcomes under optimal tolls can still be far from optimal~\cite{bilo2019dynamic, caragiannis2010taxes}.
However, recent work suggests that there is a connection between aggregate agent ``satisfaction'' and system efficiency~\cite{seaton2023intrinsic}.
Agent satisfaction can be parameterized by the \emph{stability margin}, defined as the ratio between the aggregate change in cost agents obtain by unilaterally switching from an equilibrium to a system-optimal allocation, and the system cost at equilibrium.
As a result, the stability margin captures the aggregate satisfaction agents have with their decisions at equilibrium relative to an optimal allocation.

For resource-allocation games in which payoffs are maximized, a linear program (LP) can be used to show that the PoA improves as a function of the stability margin and that, in games attaining worst-case price-of-anarchy bounds, agents are indifferent between their equilibrium actions and their actions in an optimal outcome~\cite{singh2025worst}.
This is closely related to the methods in~\cite{paccagnan2019utility, paccagnan2021optimal}, which compute the exact PoA in resource-allocation games, and can be specified for atomic congestion games with local tolls, but do not account for the stability margin.
Thus, prior to this work, it was unknown how the stability margin impacts the PoA in atomic congestion games when tolls are applied.

The main contribution of this paper is to extend the results of~\cite{singh2025worst} (which characterizes the tradeoff between the PoA and the stability margin in networked resource allocation games) to study the tradeoff between the PoA and the stability margin in atomic congestion games under local tolling mechanisms (which we refer to as \emph{incentivized} atomic congestion games).
When no toll is applied, a nonincreasing upper bound for the PoA exists as a function of the stability margin, and can be found via a $(\lambda, \mu)$-smoothness argument~\cite{seaton2023intrinsic}; however, the bound is not tight, and cannot be applied when tolls are used.
In this work, we study how the stability margin impacts the PoA under various local tolling mechanisms.
Our contributions are summarized as follows:

\begin{enumerate}
    \item First, Theorem~\ref{thm:LP_Payoff_Stability} presents an LP that bounds the PoA as a function of the stability margin for incentivized atomic congestion games.
    While it is known that the PoA \emph{weakly} improves in the stability margin, this LP provides a novel tool to investigate the problem space and verify that the PoA \emph{strictly} improves as a function of the stability margin for specific parameterizations.
    
    \item Then, Theorem~\ref{thm:Knife_Edge} shows that in incentivized atomic congestion games attaining the worst-case PoA regardless of stability margin, agents are indifferent between their equilibrium and optimal-allocation actions.
    
    \item Finally, Proposition~\ref{prop:PoA_S_LUB} shows that there exist classes of atomic congestion games without tolls such that the LP in Theorem~\ref{thm:LP_Payoff_Stability} yields a PoA bound that is strictly tighter than the previously best-known bound in~\cite{seaton2023intrinsic}.
\end{enumerate}
A high-level summary of these results is given in Fig.~\ref{fig:Main_Results}.

\begin{figure}[t]
    \vspace{2mm}
        \begin{center}
            \includegraphics[width=0.48\textwidth]{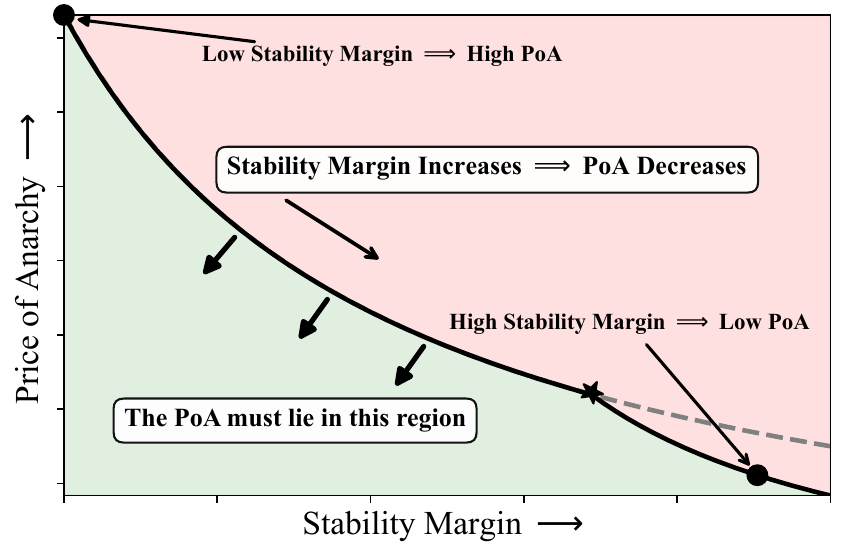}
        \end{center}
        \caption{The curve separating the red (upper-right) and green (lower-left) regions is the PoA bound parameterized by the stability margin, and is generated by our novel LP from Theorem~\ref{thm:LP_Payoff_Stability}.
        The figure shows the bound for atomic congestion games with cubic local cost functions and no tolls.
        The dots in the upper-left and lower-right corners illustrate that agents are either indifferent at inefficient equilibria (Theorem~\ref{thm:Knife_Edge}) or highly satisfied at efficient equilibria.
        The curve to the right of the star ($\star$) represents Proposition~\ref{prop:PoA_S_LUB}, which shows that the PoA bound produced by the LP strictly improves upon previously known results for specific parameterizations.
        }
        \label{fig:Main_Results}
    \vspace{-4mm}
\end{figure}

\section{Model}
We study an atomic congestion game where a set of agents \mbox{$[n] \coloneqq \{1, \ldots, n\}$} share a finite resource set~$E$.
Agent \mbox{$i \in [n]$} has a feasible action set, denoted \mbox{$\aaa_i \subseteq 2^E$}, and chooses an action from the set, \mbox{$a_i \in \aaa_i$}.
Let \mbox{$\aaa \coloneqq \aaa_1 \times \dots \times \aaa_n$} denote the joint action space, and let \mbox{$a \coloneqq (a_1, \ldots, a_n) \in \aaa$} denote an allocation of resources to agents.
For any $e \in E$ and $a \in \aaa$, denote the number of agents selecting resource $e$ in allocation~$a$ as $\vert a \vert_e \coloneqq \vert \{i \in [n] : e \in a_i \} \vert$.

Each $e \in E$ has a local cost function $\ell_e : \mathbb{N} \to \mathbb{R}_{\geq 0}$, where $\ell_e(|a|_e)$ is the cost incurred by each agent selecting resource $e$ under allocation~$a$.
We assume each resource cost function lies in $\rm{span}(\mathcal{B})$, where \mbox{$\mathcal{B} \coloneqq \{b_j : \mathbb{N} \to \mathbb{R}_{\geq 0}\}_{j \in [m]}$} denotes the set of nonnegative, nondecreasing and convex basis functions.
Here, convexity implies that \mbox{$b_j(x+1) - b_j(x)$} is nondecreasing in $x$ for all $j \in [m]$.
Thus, for every $e \in E$,
\begin{equation}
    \label{eq:Resource_Cost}
    \ell_e(\cdot) \coloneqq \sum_{j \in [m]} \alpha_{e,j} b_j(\cdot),
\end{equation}
where $\alpha_{e,j} \geq 0$ for all $j \in [m]$.
Then, the system cost is the cost incurred by all agents for a given allocation:
\begin{equation}
    \label{eq:System_Cost}
    L(a) \coloneqq \sum_{e \in E} \vert a \vert_e \ell_e(\vert a \vert_e).
\end{equation}
Intuitively, system cost represents agents' total transit time in person-hours.
To avoid trivialities, we assume~\eqref{eq:System_Cost} is strictly positive.
An instance of an atomic congestion game is then fully specified by the tuple \mbox{$G \coloneqq \leftparen [n], E, \aaa, \{\ell_e\}_{e \in E}\rightparen$}; we write $\mathcal{G}_\mathcal{B}^n$ to be the class of atomic congestion games with $n$ agents and basis functions belonging to $\mathcal{B}$.

\subsection{Tolling Mechanisms}
For a game \mbox{$G \in \mathcal{G}_{\mathcal{B}}^n$}, a \emph{local} tolling mechanism $T$ assigns each resource \mbox{$e \in E$} a toll function \mbox{$\tau_e \coloneqq T(\ell_e) : \mathbb{N} \to \mathbb{R}_{\geq 0}$} that depends only on its cost function~$\ell_e$.
A local tolling mechanism is \emph{linear} if \mbox{$T(x\ell + y\ell^\prime) = xT(\ell) + yT(\ell^\prime)$}.
Linearity does \emph{not} imply tolls are affine in the congestion level (i.e., it need \emph{not} be of the form \mbox{$\tau_e(x) = a_e x + b_e$} for $a_e, b_e \in \mathbb{R}$).
Hence, for game $G$ (whose resource cost functions are spanned by $\mathcal{B}$), any linear local tolling mechanism is determined by the toll basis functions \mbox{$\leftcurly \tau_j^\star \coloneqq T(b_j) : \mathbb{N} \to \mathbb{R}_{\geq 0} \rightcurly_{j \in [m]}$}.
Thus, for each \mbox{$e \in E$} where \mbox{$\ell_e(\cdot)$} is given by~\eqref{eq:Resource_Cost}, the resulting linear local toll is
\begin{equation}
    \label{eq:Toll_cost}
    \tau_e(\cdot) = \sum_{j \in [m]} \alpha_{e,j}\tau_j^\star(\cdot).
\end{equation}

We assume agents are selfish, and that a system designer may impose tolls on resources in $E$ to reduce the system cost in~\eqref{eq:System_Cost}.
For \mbox{$a \in \aaa$}, the cost experienced by agent \mbox{$i \in [n]$} is the sum of the local costs and tolls of their selected resources:
\begin{equation} 
    \label{eq:Agent_Cost}
    J_i(a) \coloneqq \sum_{e \in a_i} \leftparen \ell_e(\vert a \vert_e) + \tau_e(\vert a \vert_e) \rightparen.
\end{equation}
Local tolls (which depend only on information about the resource to which they are applied) are known to provide performance guarantees similar to those of global tolls (which require system-wide information)~\cite{paccagnan2021optimal}.
We focus our study on linear local tolls, since they are optimal within the class of local tolling mechanisms~\cite{paccagnan2022utility}.
We then specify an instance of an incentivized atomic congestion game with the pair~$(G,T)$.

\subsection{Nash Equilibrium}
For agent $i \in [n]$ and $a, a^\prime \in \aaa$, we write $(a^\prime_i, a_{-i})$ to denote the allocation that arises when agent $i$ unilaterally deviates from $a$ to $a^\prime$.
Now, agents choose actions that minimize their experienced costs defined in~\eqref{eq:Agent_Cost}.
This behavior is formalized by a \emph{pure Nash equilibrium}, which is guaranteed to exist for any incentivized atomic congestion game~\cite{rosenthal1973class}.

\begin{definition}
    \label{def:Nash_equilibrium}
    For an incentivized atomic congestion game $(G, T)$, $a^\Ne \in \aaa$ is a pure Nash equilibrium if for any $i \in [n]$ and any $a_i \in \aaa_i$, the following holds:
    \begin{equation}
        \label{eq:Nash_equilibrium}
        J_i(a_i^\Ne, a_{-i}^\Ne) \leq J_i(a_i, a_{-i}^\Ne).
    \end{equation}
\end{definition}
\vspace{2mm}

For game $(G,T)$, we write $\aaa^\NE(G,T) \subseteq \aaa$ to denote the set of all pure Nash equilibria, and often write $\aaa^\NE(G,T)$ simply as $\aaa^\NE$ when the context is clear.

\subsection{Performance Metrics and Research Goal}

The effectiveness of a tolling mechanism $T$ can be measured by the \emph{price of anarchy} (PoA), defined as the ratio of the highest possible system cost at Nash equilibrium relative to the system cost of an optimal allocation, taken in the worst-case over a class of games:
\begin{equation}
    \label{eq:PoA}
    \PoA \leftparen \mathcal{G}_\mathcal{B}^n, T \rightparen \coloneqq \sup_{G \in \mathcal{G}_\mathcal{B}^n} \leftparen \frac{\max_{a \in \aaa^\NE \leftparen G, T \rightparen} L(a)}{\min_{a^\ast \in \aaa} L(a^\ast)} \rightparen.
\end{equation}
Notice that \mbox{$\PoA \leftparen \mathcal{G}_\mathcal{B}^n, T \rightparen \geq 1$}.
If $\PoA \leftparen \mathcal{G}_\mathcal{B}^n, T \rightparen = 1$, then the tolls imposed by mechanism $T$ induce optimal allocations, and as $\PoA \leftparen \mathcal{G}_\mathcal{B}^n, T \rightparen$ increases, the resulting Nash equilibria can produce suboptimal allocations.

Recent work~\cite{seaton2023intrinsic} establishes a tradeoff between how ``satisfied'' agents are with their decisions at a Nash equilibrium (relative to their decisions at a system optimal allocation) and the quality of the system cost at that equilibrium.
For incentivized game $(G,T)$, let \mbox{$\aaa^\WNE \coloneqq \arg\max_{a \in \aaa^\NE} L(a)$} denote the set of worst-case Nash equilibria, and let \mbox{$\aaa^\opt \coloneqq \arg\min_{a \in \aaa} L(a)$} denote the set of allocations that minimize system cost.
Then, aggregate agent satisfaction can be measured with the \emph{stability margin}, defined as the aggregate worst-case difference between each agent's cost if they unilaterally deviate to an optimal action from an equilibrium action, relative to the system cost at equilibrium:
\begin{equation}
    \label{eq:Stability_Margin}
    \mathcal{S} (G, T) \coloneqq\hspace{-.75mm} \max_{a \in \aaa^\WNE} \max_{a^\opt \in \aaa^\opt} \frac{\sum_{i \in [n]} J_i(a^\opt_i, a_{-i}) - J_i(a)}{L(a)}.
\end{equation}

Since users minimize their individual cost at equilibrium, $\mathcal{S} (G, T) \geq 0$.
If $\mathcal{S} (G, T) = 0$, then all agents are indifferent between their equilibrium action and their action at a system optimal allocation, and as $\mathcal{S} (G, T)$ increases, agents increasingly prefer their equilibrium action.
Thus, when $\mathcal{S} (G, T)$ is small, agents require little incentive to deviate from worst-case equilibria to system optimal outcomes.

For a class of games $\mathcal{G}_{\mathcal{B}}^n$, tolling mechanism $T$, and scalar $S \geq 0$, we denote $\mathcal{G}_{\mathcal{B}}^n (S, T) \coloneqq \leftcurly G \in \mathcal{G}_{\mathcal{B}}^n \;\middle|\; \mathcal{S}(G,T)  \geq S \rightcurly$
as the class of incentivized games such that every $(G,T) \in \mathcal{G}_{\mathcal{B}}^n (S, T)$ has stability margin greater than or equal to~$S$.
Since the stability margin is nonnegative, $\mathcal{G}_{\mathcal{B}}^n$ is the class of games such that every $G \in \mathcal{G}_{\mathcal{B}}^n$ has stability margin greater than or equal to $0$, so we write $\mathcal{G}_{\mathcal{B}}^n$ as $\mathcal{G}_{\mathcal{B}}^n (0, T)$ for clarity.
We also write $\PoA(\mathcal{G}_{\mathcal{B}}^n(S,T), T)$ simply as $\PoA(\mathcal{G}_{\mathcal{B}}^n(S,T))$.

Finally, we are interested in studying the tradeoff between system-level efficiency at worst-case equilibria and the extent to which agents are satisfied with those equilibria:
\begin{problem}
    \label{prob:Main}
    Given class of games $\mathcal{G}_\mathcal{B}^n(S, T)$, identify an upper bound on $\PoA \leftparen \mathcal{G}_\mathcal{B}^n(S, T) \rightparen$ as a function of~$S$.
\end{problem}

\section{Main Results}
This section provides a principled method for solving Problem~\ref{prob:Main} and uses it to investigate the tradeoffs between the PoA and the stability margin.
Specifically:
\begin{enumerate}
    \item Theorem~\ref{thm:LP_Payoff_Stability} provides an LP that solves Problem~\ref{prob:Main}; for a given class of games $\mathcal{G}_{\mathcal{B}}^n(S,T)$, the LP produces a bound $p^\star$ satisfying \mbox{$\PoA \leftparen \mathcal{G}_{\mathcal{B}}^n(S,T) \rightparen \leq 1/p^\star$}.
    
    \item Theorem~\ref{thm:Knife_Edge} shows that if a game $(G,T) \in \mathcal{G}_{\mathcal{B}}^n(0,T)$ attains $\PoA \leftparen \mathcal{G}_{\mathcal{B}}^n(0,T) \rightparen$, then all agents are indifferent between their actions at a worst-case equilibrium and an optimal allocation, implying $\mathcal{S} (G,T) = 0$.
    
    \item Proposition~\ref{prop:PoA_S_LUB} applies the LP from Theorem~\ref{thm:LP_Payoff_Stability} to classes of games without tolls, $\mathcal{G}_{\mathcal{B}}^n(S,0)$, and yields a PoA bound that strictly improves upon known bounds for specific values of $S$.
\end{enumerate}

\subsection{Tradeoff Between Price of Anarchy and Stability Margin}

Our first main result presents a linear program (LP) that solves Problem~\ref{prob:Main}.
The LP does this by extending the technical work presented in~\cite{paccagnan2021optimal} in line with the approach presented in~\cite{singh2025worst}.
In~\cite{paccagnan2021optimal}, it is shown that for any class of atomic congestion games $\mathcal{G}_{\mathcal{B}}^n(0, T)$, the \emph{exact} price of anarchy $\PoA \leftparen \mathcal{G}_\mathcal{B}^n (0, T) \rightparen$ can be obtained from an LP defined on the reduced class of games, $\bar{\mathcal{G}}_{\mathcal{B}}^n(0, T)$, where each agent is restricted to two actions: their worst-case equilibrium action and their action in an optimal allocation.
These results are extended in~\cite{singh2025worst}, which presents an LP that generates a \emph{lower} bound on the PoA as a function of $S \geq 0$ for resource allocation games where \emph{payoffs} are \emph{maximized}.
Our work extends this result by presenting an LP that generates an \emph{upper} bound on the PoA as a function of $S \geq 0$ for resource allocation games where \emph{costs} are \emph{minimized}.

The LPs presented in this work have decision variables that are indexed by the following set:
\begin{align}
    \nonumber
    \mathcal{I} \coloneqq \Big\{ 
    &\leftparen ((x_1, y_1, z_1), \ldots, (x_n, y_n, z_n)), j \rightparen : j \in [m],
    \\ \nonumber
    &(x_i, y_i, z_i) \in \{0, 1\}^3 \quad \forall i \in [n],
    \\ \nonumber
    &0 \leq x_i + y_i + z_i \leq 1 \quad \forall i \in [n], 
    \\ \label{eq:I}
    &1 \leq \sum_{i \in [n]} (x_i + y_i + z_i) \leq n
    \Big\}.
\end{align}
For ease of notation, let $(t, j) = ((t_1, \ldots, t_n), j)$, where \mbox{$t_i = (x_i, y_i, z_i)$} for all $i \in [n]$, so $(t, j)$ is a typical element of~$\mathcal{I}$.
Intuitively, $\mathcal{I}$ is the set of all feasible ways each agent can contribute to resource usage in the equilibrium allocation only ($y$), the optimal allocation only ($z$), or both ($x$).
Thus, $\mathcal{I}$ is exponentially large in $n$.
The index $j \in [m]$ allows the LP to search over coefficients for each basis function at each optimal and worst-case equilibrium resource.
For any $t$, let
\begin{align}
    \label{eq:Y_t_and_Z_t}
    Y_t \coloneqq \sum_{i \in [n]} (x_i + y_i), 
    \text{ and }
    Z_t \coloneqq \sum_{i \in [n]} (x_i + z_i).
\end{align}
Intuitively, $Y_t$ is the aggregate number of agents using a resource at equilibrium, and $Z_t$ is the aggregate number of agents using a resource in an optimal allocation.

\begin{theorem}
    \label{thm:LP_Payoff_Stability}
    Consider class of incentivized atomic congestion games $\mathcal{G}_{\mathcal{B}}^n (S, T)$.
    Then, it holds that 
    \begin{equation}
        \label{eq:PoA_S_UB}
        \PoA \leftparen \mathcal{G}_{\mathcal{B}}^n (S, T) \rightparen \leq 1/p^\star,
    \end{equation}
    where $p^\star$ is the value of the following LP:
    \begin{subequations}
        \label{eq:LP_S}
        \begin{align}
            \label{eq:OPT_cond_S}
            p^\star = \min_{\{\theta(t, j)\}} \quad &\sum_{(t, j) \in \mathcal{I}} Z_t b_j(Z_t) \theta(t, j)
            \\ \label{eq:NE_cond_S}
            \text{s.t.}
            \quad \sum_{(t, j) \in \mathcal{I}} &\leftbracket y_if_j(Y_t) - z_if_j(Y_t+1) \rightbracket \theta(t, j) \leq 0 \; \forall i \in [n],
            \\ \label{eq:STAB_cond_S}
            \sum_{i \in [n]} \sum_{(t, j) \in \mathcal{I}} &\leftbracket y_if_j(Y_t) - z_if_j(Y_t+1) \rightbracket \theta(t, j) \leq -S,
            \\ \label{eq:NORM_cond_S}
            &\sum_{(t, j) \in \mathcal{I}} Y_t b_j(Y_t) \theta(t, j)=1
            \\ \label{eq:COEFF_cond_S}
            & \theta(t, j) \geq 0, \quad \forall (t, j) \in \mathcal{I},
        \end{align}
    \end{subequations}
    where $f_j \coloneqq b_j + \tau^\star_j : \mathbb{N} \to \mathbb{R}_{\geq 0}$, and for convenience, we define $b_j(0) = f_j(0) = f_j(n+1) = 0$ for each $j \in [m]$.
\end{theorem}
\vspace{2mm}

We provide a brief interpretation before presenting the proof of Theorem~\ref{thm:LP_Payoff_Stability}.
For a class of games $\mathcal{G}_{\mathcal{B}}^n(S,T)$, the LP in~\eqref{eq:LP_S} generates an upper bound on $\PoA \leftparen \mathcal{G}_{\mathcal{B}}^n(S,T) \rightparen$ that is equivalent to upper-bounding the solution to the following:
\begin{equation}
    \label{eq:Optimization_Prob}
    \sup_{(G,T) \in \mathcal{G}_\mathcal{B}^n(S, T)} \frac{\max_{a^\Ne \in \aaa^\NE} L(a^\Ne)}{\min_{a^\ast \in \aaa} L(a^\ast)}.
\end{equation}
Specifically, the normalization constraint in~\eqref{eq:NORM_cond_S} sets the system cost at a worst-case equilibrium to one. As a result, minimizing the LP objective in~\eqref{eq:OPT_cond_S}, which represents the system cost of an optimal allocation, is equivalent to maximizing the ratio in~\eqref{eq:Optimization_Prob}. The constraints in~\eqref{eq:NE_cond_S} and~\eqref{eq:STAB_cond_S} enforce the equilibrium and stability-margin conditions, respectively, while~\eqref{eq:COEFF_cond_S} parameterizes the coefficients of the resource-cost functions in a worst-case game.

Tuning $S$ decides which game instances from $\mathcal{G}_{\mathcal{B}}^n (S,T)$ are encoded as variables in the LP.
As $S$ increases, fewer games satisfy the stability-margin constraint in~\eqref{eq:STAB_cond_S}, so no new games arise in $\mathcal{G}_{\mathcal{B}}^n(S,T)$ relative to $\mathcal{G}_{\mathcal{B}}^n(0,T)$.
Thus, the bound on $\PoA \leftparen \mathcal{G}_{\mathcal{B}}^n (S, T) \rightparen$ is nonincreasing as a function of~$S$.
However, by performing a parameter sweep on $S$, we can explore the tradeoff between the PoA and the stability margin for specific tolling mechanisms, and determine that the bound the LP provides for $\PoA \leftparen \mathcal{G}_{\mathcal{B}}^n (S, T) \rightparen$ is often \emph{strictly} decreasing as a function of $S$.
Numerical experiments that illustrate these results are presented in Section~\ref{sec:Experiments}.
We now proceed with the proof of Theorem~\ref{thm:LP_Payoff_Stability}.

Theorem~\ref{thm:LP_Payoff_Stability} extends~\cite[Theorem 1]{paccagnan2021optimal}, where some parts of our proof align with their proof.
For clarity, we provide our complete proof, which is completed in three steps:
\begin{enumerate}
    \item First, we observe that the PoA over a class of games $\mathcal{G}_{\mathcal{B}}^n (S, T)$ is the same as the PoA over the reduced class of games $\bar{\mathcal{G}}_{\mathcal{B}}^n (S,T)$, where the action set of every agent is reduced to an action set containing only two actions.
    
    \item Next, we show $\PoA \leftparen \bar{\mathcal{G}}_{\mathcal{B}}^n (S, T) \rightparen$ is unchanged when the system cost is constrained to be equal to $1$ at a worst-case equilibrium of any $(\bar{G},T) \in \bar{\mathcal{G}}_{\mathcal{B}}^n (S, T)$; i.e., \mbox{$L(a^\Ne) = 1$} for any \mbox{$a^\Ne \in \arg\max_{a \in \aaa^\NE(\bar{G},T)} L(a)$}.

    \item Finally, we show how to parameterize the LP by introducing the variables $\theta(t, j)$, where $(t, j) \in \mathcal{I}$. 
\end{enumerate}

\textit{Proof of Theorem~\ref{thm:LP_Payoff_Stability}:}
Let $\mathcal{G}_{\mathcal{B}}^n (S, T)$ be a class of atomic congestion games with stability margin no less than \mbox{$S \geq 0$} and tolling mechanism $T$.
Define $\bar{\mathcal{G}}_{\mathcal{B}}^n(S, T)$ to be the reduced class of games where each agent has access to two actions.

\textit{Step 1:}
We first show that  
\begin{equation}
    \label{eq:PoA_reduced}
    \PoA \leftparen \mathcal{G}_{\mathcal{B}}^n (S, T) \rightparen = \PoA \leftparen \bar{\mathcal{G}}_{\mathcal{B}}^n (S, T) \rightparen.  
\end{equation}
Consider a game $(G,T) \in \mathcal{G}_{\mathcal{B}}^n (S, T)$ such that $(G,T)$ achieves the PoA bounds for $\mathcal{G}_{\mathcal{B}}^n (S, T)$.
That is, there exists $a^\Ne \in \aaa^\NE(G,T)$ and $a^\opt \in \aaa^\opt$ such that 
\begin{equation}
    \label{eq:proof_1}
    \frac{L(a^\Ne)}{L(a^\opt)} = \PoA \leftparen \mathcal{G}_{\mathcal{B}}^n (S, T) \rightparen.
\end{equation}
Now, construct game $\bar{G}$ to be identical to $G$ in everything but the action sets.
For all $i \in [n]$, the action sets of $\bar{G}$ are  \mbox{$\bar{\aaa}_i \coloneqq \{a^\Ne_i, a^\opt_i\} \subseteq \aaa_i$}, so the action set for each agent is restricted to their action in $a^\Ne$ and $a^\opt$ only.
Since actions were only removed from $G$ to $\bar{G}$, the cost functions $\{J_i\}_{i \in [n]}$ are well-defined and attain equivalent values for $\bar{\aaa}$ in $\bar{G}$ as they do for $\bar{\aaa}$ in $G$.
Thus, $a^\Ne$ remains an equilibrium allocation and $a^\opt$ remains an optimal allocation in $\bar{G}$.
Further, since each agent in $\bar{G}$ is restricted to their actions in $a^\Ne$ and $a^\opt$, $\mathcal{S}(\bar{G}, T) \geq S$ and $(\bar{G}, T) \in \bar{\mathcal{G}}_{\mathcal{B}}^n (S, T)$.
Thus,
\begin{equation}
    \label{eq:proof_2}
    \PoA \leftparen \bar{\mathcal{G}}_{\mathcal{B}}^n (S, T) \rightparen \geq \frac{L(a^\Ne)}{L(a^\opt)}.
\end{equation}
Since $\bar{\mathcal{G}}_{\mathcal{B}}^n(S, T)$ is defined by games in $\mathcal{G}_{\mathcal{B}}^n(S, T)$ that are restricted by their action set, it follows that \mbox{$\bar{\mathcal{G}}_{\mathcal{B}}^n(S, T) \subseteq \mathcal{G}_{\mathcal{B}}^n(S, T)$}, so that 
\begin{equation}
    \label{eq:proof_3}
    \PoA \leftparen \bar{\mathcal{G}}_{\mathcal{B}}^n (S, T) \rightparen \leq \PoA \leftparen \mathcal{G}_{\mathcal{B}}^n (S, T) \rightparen.
\end{equation}
Combining~\eqref{eq:proof_1},~\eqref{eq:proof_2}, and~\eqref{eq:proof_3} implies that~\eqref{eq:PoA_reduced} holds.

\textit{Step 2:}
Next, for any game $(\bar{G},T) \in \bar{\mathcal{G}}_{\mathcal{B}}^n (S, T)$ that achieves the PoA bound, we construct a corresponding game $\bar{G}^\prime$ in the following manner.
For each resource $e \in E$, define the local cost in $\bar{G}^\prime$ to be $\ell_e(\cdot)/L(a^\Ne)$.
Since $T$ is linear, the corresponding toll in $\bar{G}^\prime$ is $\tau_e(\cdot)/L(a^\Ne)$, so the cost experienced by every agent is scaled by the same positive constant $1/L(a^\Ne)$. Hence, $a^\Ne$ remains a worst-case Nash equilibrium in $\bar{G}^\prime$, $a^\opt$ remains an optimal allocation in $\bar{G}^\prime$, and $\mathcal{S}(\bar{G}^\prime,T)=\mathcal{S}(\bar{G},T)\geq S$.
Thus, $a^\Ne$ generates a system cost of $L(a^\Ne) / L(a^\Ne) = 1$ for $\bar{G}^\prime$.
Since this procedure scales the system cost for all allocations by the same coefficient, $a^\Ne$ remains the worst-case Nash equilibrium and $a^\opt$ remains an optimal allocation for $\bar{G}^\prime$.
Thus, $(\bar{G}^\prime, T)$ attains the PoA bound for the class of games $\bar{\mathcal{G}}_{\mathcal{B}}^n (S, T)$.

\textit{Step 3:}
By \textit{Step 1} and \textit{Step 2}, $\PoA \leftparen \mathcal{G}_{\mathcal{B}}^n (S, T) \rightparen$ is upper-bounded by the reciprocal of the following problem:
\begin{align}
    \nonumber
    \inf_{(\bar{G},T) \in \bar{\mathcal{G}}_{\mathcal{B}}^n (S, T)} &L(a^\opt)
    \\ \nonumber
    \text{s.t.} \quad &J_i(a_i^\Ne, a_{-i}^\Ne) \leq J_i(a_i^\opt, a_{-i}^\Ne) \; \forall i \in [n],
    \\ \nonumber
    &S \leq \frac{\sum_{i \in [n]} J_i(a^\opt_i, a^\Ne_{-i}) - J_i(a^\Ne)}{L(a^\Ne)}
    \\ \label{eq:L_star_opt}
    &L(a^\Ne) = 1.
\end{align}
Thus, our final step is to show that $p^\star$, the value of the LP in~\eqref{eq:LP_S}, is the solution to~\eqref{eq:L_star_opt}.
To do so, we introduce the variables $\theta(t,j)$, where $(t,j) \in \mathcal{I}$.
Recall, for all $i \in [n]$, the feasible set of actions is $\bar{\aaa}_i = \{a^\Ne_i, a^\opt_i \}$.
Let
\begin{align}
    \nonumber
    x_{e,i} &\coloneqq \vert \{ e \in a_i^\Ne \cap a_i^\opt \} \vert,
    \\ \nonumber
    y_{e,i} &\coloneqq \vert \{ e \in a_i^\Ne \setminus a_i^\opt \} \vert,
    \\ \nonumber
    z_{e,i} &\coloneqq \vert \{ e \in a_i^\opt \setminus a_i^\Ne \} \vert,
\end{align}
where $\vert \cdot \vert$ denotes cardinality.
Thus, for each tuple $(t,j) \in \mathcal{I}$, where $t = (t_1, \ldots, t_n)$, and $t_i = (x_i, y_i, z_i)$ for $i \in [n]$, we may write $x_{e,i} = x_i$, $y_{e,i} = y_i$, and $z_{e,i} = z_i$.
Define $E(t)$ to be the set containing all resources $e \in E$ that are selected by $Y_t = \sum_{i \in [n]} (x_i + y_i)$ agents at the equilibrium and selected by $Z_t = \sum_{i \in [n]} (x_i + z_i)$ agents at the optimum.
Hence,
\begin{equation}
    E(t) \coloneqq \leftcurly
    e \in E :
    \sum_{i \in [n]} (x_i + y_i) = Y_t,
    \sum_{i \in [n]} (x_i + z_i) = Z_t
    \rightcurly.
\end{equation}
Then, for each $(t, j) \in \mathcal{I}$, we define $\theta(t,j) \geq 0$ to be the sum of the coefficients for the resource costs in $E(t)$ and for basis function $b_j \in \mathcal{B}$, i.e., $\theta(t,j) \coloneqq \sum_{e \in E(t)} \alpha_{e,j}$.
Thus, using the notation just introduced, we can rewrite the system cost at $a^\opt$ and $a^\Ne$ in~\eqref{eq:L_star_opt}, respectively,
\begin{align}
    \label{eq:transform_1}
    L(a^\opt) &= \sum_{e \in E} \vert a^\opt \vert_e \ell_e(\vert a^\opt \vert_e)
    = \sum_{(t, j) \in \mathcal{I}} Z_t b_j(Z_t)\theta(t, j),
    \\ \label{eq:transform_2}
    L(a^\Ne) &= \sum_{e \in E} \vert a^\Ne \vert_e \ell_e(\vert a^\Ne \vert_e)
    = \sum_{(t, j) \in \mathcal{I}} Y_t b_j(Y_t) \theta(t, j).
\end{align}
Now, for each $i \in [n]$, the individual user cost difference $\Delta_i \coloneqq J_i(a_i^\Ne, a_{-i}^\Ne) - J_i(a_i^\opt, a_{-i}^\Ne)$, becomes
\begin{align}
    \nonumber
    \Delta_J = &\sum_{e \in a_i^\Ne} \leftparen \ell_e(\vert a^\Ne \vert_e) + \tau_e(\vert a^\Ne \vert_e) \rightparen 
    \\ \nonumber
    &- \sum_{e \in a_i^\opt} \leftparen \ell_e(\vert a_{-i}^\Ne\vert_e + 1) + \tau_e(\vert a_{-i}^\Ne \vert_e + 1) \rightparen
    \\ \label{eq:transform_3}
    = &\sum_{(t, j) \in \mathcal{I}} \leftbracket y_if_j(Y_t) \rightbracket \theta(t, j) - \hspace{-2mm}\sum_{(t, j) \in \mathcal{I}} \leftbracket z_i f_j(Y_t+1) \rightbracket \theta(t, j).
\end{align}
By substituting the values in~\eqref{eq:transform_1}--\eqref{eq:transform_3} appropriately into~\eqref{eq:L_star_opt} and restricting $\theta(t,j)$ to be nonnegative for all $(t,j) \in \mathcal{I}$, the variables $\theta(t,j)$ are sufficient to characterize the upper bound on $\PoA \leftparen \mathcal{G}_{\mathcal{B}}^n (S,T) \rightparen$ via the LP in~\eqref{eq:LP_S}. \hfill $\blacksquare$

\subsection{Individual Agents are Indifferent Between Worst-Case Equilibria and Optimal Outcomes}

A system designer may ask how satisfied agents are with their actions when the PoA is unconstrained by the stability margin.
Theorem~\ref{thm:Knife_Edge} shows that in games attaining the worst-case PoA bound, every agent is indifferent between their equilibrium action and their action in an optimal allocation.

\begin{theorem}
    \label{thm:Knife_Edge}
    Consider class of incentivized atomic congestion games $\mathcal{G}_{\mathcal{B}}^n(0, T)$.
    Let \mbox{$(G,T) \in \mathcal{G}_{\mathcal{B}}^n(0, T)$} be any game that achieves $\PoA \leftparen \mathcal{G}_\mathcal{B}^n(0, T) \rightparen$.
    That is,
    \begin{equation}
        \label{eq:PoA_worst_case}
        \frac{L(a^\Ne)}{L(a^\opt)} = \PoA \leftparen \mathcal{G}_\mathcal{B}^n(0, T) \rightparen,
    \end{equation}
    for some Nash equilibrium \mbox{$a^\Ne \in \aaa^\NE$} and optimal allocation $a^\opt \in \aaa^\opt$ of $G$. 
    Then, the following holds:
    \begin{equation}
        \label{eq:Knife_Edge}
        J_i(a_i^\Ne, a_{-i}^\Ne) = J_i(a_i^\opt, a_{-i}^\Ne) \quad \forall i \in [n].
    \end{equation}
    Moreover, the stability margin satisfies $\mathcal{S} (G,T) = 0$.
\end{theorem}
\vspace{2mm}

Before providing the proof of Theorem~\ref{thm:Knife_Edge}, we provide a brief discussion and supporting material.
Intuitively, for a game $(G,T)$ that attains the worst-case bound in~\eqref{eq:PoA_worst_case}, any agent can deviate from their action at equilibrium to their action in an optimal allocation without increasing their individual cost.
Since the equilibrium and optimal allocations are arbitrary, it follows that $\mathcal{S} (G,T) = 0$; however, this result is not trivial as $(G,T) \in \mathcal{G}_\mathcal{B}^n(0, T)$ only implies $\mathcal{S} (G,T) \geq 0$.

Theorem~\ref{thm:Knife_Edge} also has real-world implications, as it appears that fine-tuning is required in order for games with worst-case efficiency to arise, meaning their occurrence is unlikely.
Further, even in the presence of such games, worst-case equilibria will be easy to avoid by allowing agents to randomize over their actions at equilibria.
Next, for the class of incentivized atomic congestion games $\mathcal{G}_{\mathcal{B}}^n(0, T)$, Lemma~\ref{lem:OPT_NE_equality} shows that the constraint in~\eqref{eq:NE_cond_S} holds with equality.

\begin{lemma}
    \label{lem:OPT_NE_equality}
    Let $\mathcal{G}_{\mathcal{B}}^n(0, T)$ be a class of atomic congestion games.
    In any optimal solution to the LP in~\eqref{eq:LP_S} over $\mathcal{G}_{\mathcal{B}}^n(0, T)$, the constraint in~\eqref{eq:NE_cond_S} is tight for every \mbox{$i \in [n]$}.
    
    \begin{proof}
        When $S=0$, any solution that satisfies the constraint in~\eqref{eq:NE_cond_S} immediately satisfies the constraint in~\eqref{eq:STAB_cond_S}.
        Thus, the LP in~\eqref{eq:LP_S} can be rewritten as
        \begin{equation}
            \label{eq:PoA_tight_bound}
            \PoA \leftparen \mathcal{G}_{\mathcal{B}}^n(0, T) \rightparen = 1/q^\star,
        \end{equation}
        where $q^\star$ is the value of the following LP:
        \begin{subequations}
            \label{eq:LP}
            \begin{align}
                \label{eq:OPT_cond}
                q^\star = \min_{\theta(t, j)} \quad &\sum_{(t, j) \in \mathcal{I}} Z_t b_j(Z_t) \theta(t, j)
                \\ \label{eq:NE_cond}
                \text{s.t.}
                \quad \sum_{(t, j) \in \mathcal{I}} &\leftbracket y_if_j(Y_t) - z_if_j(Y_t+1) \rightbracket \theta(t, j) \leq 0 \; \forall i \in [n],
                \\ \label{eq:NORM_cond}
                &\sum_{(t, j) \in \mathcal{I}} Y_t b_j(Y_t) \theta(t, j)=1,
                \\ \label{eq:COEFF_cond}
                & \theta(t, j) \geq 0 \quad \forall (t, j) \in \mathcal{I},
            \end{align}
        \end{subequations}
        where $f_j \coloneqq b_j + \tau^\star_j : \mathbb{N} \to \mathbb{R}_{\geq 0}$, and for convenience, we define $b_j(0) = f_j(0) = f_j(n+1) = 0$ for each $j \in [m]$.
        \vspace{2mm}
        \\Thus, the LP in~\eqref{eq:LP} determines the PoA for a class of games that is unrestricted by the stability margin.
        It is also known that~\eqref{eq:PoA_tight_bound} holds with equality (for a formal proof, see~\cite[Theorem 2, Lemma 3]{paccagnan2019utility}).
        Now, the dual of the LP in~\eqref{eq:LP} is
        \begin{align}
            \label{eq:LP_dual}
            &d^\star = \max_{\lambda_1, \ldots \lambda_n \geq 0, \mu \in \mathbb{R}} \mu \quad \text{subject to}
            \\ \nonumber
            &Z_t b_j(Z_t) + \sum_{i \in [n]} \lambda_i \leftparen y_i f_j(Y_t) - z_i f_j(Y_t+1) \rightparen \geq \mu Y_t b_j(Y_t)
            \\ \nonumber
            &\forall (t,j)\in\mathcal{I},
        \end{align}
        and $b_j(0) = f_j(0) = f_j(n+1) = 0$ for each $j \in [m]$.
        \vspace{2mm}
        \\ Next, fix any $j \in [m]$ and $i \in [n]$, and consider the case that $y_i = 1$ and $x_i = z_i = x_l = y_l = z_l = 0$ for all $l \neq i$.
        Then, the constraint in the dual formulation in~\eqref{eq:LP_dual} reduces to
        \begin{equation}
            \label{eq:OPT_NE_equality_part_1}
            \lambda_i f_j(1) \geq \mu b_j(1) > 0,
        \end{equation}
        where the strict inequality in~\eqref{eq:OPT_NE_equality_part_1} follows from strong duality of linear programs~\cite{boyd2004convex}, as $d^\star = q^\star > 0$ (where $q^\star$ is given by~\eqref{eq:LP}) and $b_j(1)$ is strictly positive (otherwise, it can be shown that $\PoA \leftparen \mathcal{G}_{\mathcal{B}}^n(0, T) \rightparen \leq 0$, a contradiction), so it must be the case that $\mu b_j(1) > 0$.
        It also follows that $f_j(1) > 0$ since \mbox{$f_j(1) \geq b_j(1)$}, so any feasible $\lambda_i > 0$.
        Now, at the dual optimum, $\lambda^\star_i > 0$ for all \mbox{$i \in [n]$}.
        Hence, by complementary slackness~\cite{boyd2004convex}, the corresponding primal constraints in~\eqref{eq:NE_cond} are tight at the primal optimum.
        Thus, if $\theta^\star(t,j)$ is the optimal solution to the LP in~\eqref{eq:LP}, then
        $\sum_{(t, j) \in \mathcal{I}} \leftbracket y_if_j(Y_t) - z_if_j(Y_t+1) \rightbracket \theta^\star(t, j) = 0 \; \forall i \in [n]$.
    \end{proof}
\end{lemma}
\vspace{2mm}

\textit{Proof of Theorem~\ref{thm:Knife_Edge}:}
Recall, $(G,T) \in \mathcal{G}_{\mathcal{B}}^n(0, T)$ attains $\PoA \leftparen \mathcal{G}_\mathcal{B}^n(0, T) \rightparen$, so there exist \mbox{$a^\Ne \in \aaa^\NE$} and \mbox{$a^\opt \in \aaa^\opt$} such that~\eqref{eq:PoA_worst_case} holds.
If each agent has no more than two available actions, then~\eqref{eq:Knife_Edge} holds trivially or by Lemma~\ref{lem:OPT_NE_equality}.
Thus, assume \mbox{$\vert \aaa_i \vert > 2$} for all \mbox{$i \in [n]$}.
Following the same reduction used in \textit{Step~1} of the proof of Theorem~\ref{thm:LP_Payoff_Stability}, construct game $\bar{G} \coloneqq \leftparen [n], E, \bar{\aaa}, \{\ell_e\}_{e \in E}\rightparen$ that coincides with $G$, except each agent's action set is restricted to their actions in $a^\Ne$ and $a^\opt$ in $\bar{G}$ only, i.e., \mbox{$\bar{\aaa}_i \coloneqq \{a^\Ne_i, a^\opt_i\} \subseteq \aaa_i$}.
By Lemma~\ref{lem:OPT_NE_equality}, since each agent has two actions in $\bar{G}$,~\eqref{eq:Knife_Edge} holds for $\bar{G}$.
Further, the cost functions are equivalent in $\bar{G}$ and $G$, and $a^\Ne$ is any Nash equilibrium and $a^\opt$ is any optimal allocation that achieves the PoA bound, so~\eqref{eq:Knife_Edge} also holds for $G$.
By summing~\eqref{eq:Knife_Edge} over $i \in  [n]$, $\mathcal{S}(G,T) = 0$. \hfill $\blacksquare$

\subsection{The PoA for Atomic Congestion Games without Tolls}
Finally, we demonstrate that in games without tolls, the LP in~\eqref{eq:LP_S} yields an upper bound on the PoA as a function of the stability margin that strictly improves upon known bounds for certain parameter regimes.
In an allocation \mbox{$a \in \aaa$}, denote the cost experienced by agent \mbox{$i \in [n]$} when no tolls are applied as $C_i(a) \coloneqq \sum_{e \in a_i} \ell_e(\vert a \vert_e)$.
Many results produce a common upper bound on the PoA for atomic congestion games when no tolls are applied (e.g.,~\cite{roughgarden2002bad, awerbuch2005price, perakis2007price, hoeksma2011price}); these methods can be unified by the notion of $(\lambda, \mu)$-smoothness~\cite{roughgardenIntrinsicRobustnessPrice2015}:
\begin{definition}
    \label{def:Smooth_Game}
    A cost-minimization game is $(\lambda, \mu)$-smooth if for every two allocations $a$ and $a^\prime$,
    \begin{equation}
        \sum_{i \in [n]} C_i(a^\prime_i, a_{-i}) \leq \lambda L(a^\prime) + \mu L(a).
    \end{equation}
\end{definition}
\vspace{2mm}

For any class of games $\mathcal{G}_\mathcal{B}^n(0, 0)$ that is characterized by Definition~\ref{def:Smooth_Game}, $\PoA \leftparen \mathcal{G}_\mathcal{B}^n(0, 0) \rightparen \leq \frac{\lambda}{1 - \mu}$~\cite{roughgardenIntrinsicRobustnessPrice2015}.
This bound was extended in~\cite{seaton2023intrinsic} to upper-bound the PoA for $\mathcal{G}_\mathcal{B}^n(S, 0))$:
\begin{equation}
    \label{eq:PoA_S_Seaton}
    \PoA \leftparen \mathcal{G}_\mathcal{B}^n(S, 0) \rightparen \leq \leftcurly 1, \frac{\lambda}{1 - \mu + S} \rightcurly.
\end{equation}
Our final result shows that, for certain classes of atomic congestion games with no tolls, the LP in~\eqref{eq:LP_S} of Theorem~\ref{thm:LP_Payoff_Stability} yields a strictly tighter PoA bound than~\eqref{eq:PoA_S_Seaton}.

\begin{proposition}
    \label{prop:PoA_S_LUB}
    For some $S \geq 0$, there exists a class of atomic congestion games $\mathcal{G}_{\mathcal{B}}^n(S, 0)$ for which \mbox{$P^\star \coloneqq 1/p^\star$}, where $p^\star$ is the optimal value of the LP in~\eqref{eq:LP_S} when applied to  $\mathcal{G}_{\mathcal{B}}^n (S, 0)$ (i.e., $\PoA \leftparen \mathcal{G}_{\mathcal{B}}^n (S, 0) \rightparen \leq P^\star$), satisfies
    \begin{equation}
        \label{eq:PoA_S_LUB}
        P^\star < \max \leftcurly 1, \frac{\lambda}{1 - \mu + S} \rightcurly.
    \end{equation}
\end{proposition}
\vspace{2mm}

Proposition~\ref{prop:PoA_S_LUB} is proved by verifying~\eqref{eq:PoA_S_LUB} for specific classes of games, which is done in Section~\ref{ssec:PoA_NoTolls} and Fig.~\ref{fig:PoA_vs_Stability_No_Tolls}.

\section{Experimental Results}\label{sec:Experiments}

\subsection{Comparison of Tolling Mechanisms}

Here, we numerically evaluate the PoA bound provided by Theorem~\ref{thm:LP_Payoff_Stability}.
We solve the LP in~\eqref{eq:LP_S} for three linear local tolling mechanisms (in addition to the no-toll case) to illustrate that the PoA is strictly decreasing as a function of the stability margin.
The following tolls are considered:
\begin{enumerate}
    \item Optimal local tolls~\cite{paccagnan2021optimal} determine the congestion \textit{dependent} toll that minimizes $\PoA \leftparen \mathcal{G}_{\mathcal{B}}^n(0, T) \rightparen$.
    
    \item Optimal constant local tolls~\cite{paccagnan2021optimal} determine the congestion \textit{independent} toll that minimizes $\PoA \leftparen \mathcal{G}_{\mathcal{B}}^n(0, T) \rightparen$
    
    \item Marginal-cost tolls~\cite{sandholm2007pigouvian} charge each agent for the negative externality they impose on others utilizing the same resource.
    Thus, for $e \in E$ with $x \coloneqq \vert a \vert_e$, the toll \mbox{$\tau_e(x) = (x-1)\leftparen \ell_e(x) - \ell_e(x-1) \rightparen$} is applied.
\end{enumerate}

Fig.~\ref{fig:PoA_vs_Stability_Comparing_Tolls} plots the tradeoff between the PoA and the stability margin for these tolling mechanisms by sweeping $S$ from $0$ to $1$ and solving the LP in~\eqref{eq:LP_S} for each $S$ and $T$.
The LP is solved for games with cubic latency functions (i.e., \mbox{$\mathcal{B} = \{1, x, x^2, x^3\}$}).
For each tolling mechanism, the upper bound produced by the LP does not exceed the PoA bound reported in~\cite[Table 1]{paccagnan2021optimal} (shown here as the far-left side of the plot at \mbox{$S=0$}) because increasing the stability margin only restricts which games are encoded in the LP.
Because these tolls require no information about the stability margin, a system designer can apply them with no knowledge of agent satisfaction and still guarantee that efficiency strictly improves upon the worst-case PoA bounds (i.e., when \mbox{$S=0$}) as agents become more satisfied at equilibrium.

\begin{figure}[t]
    \vspace{2mm}
        \begin{center}
            \includegraphics[width=0.48\textwidth]{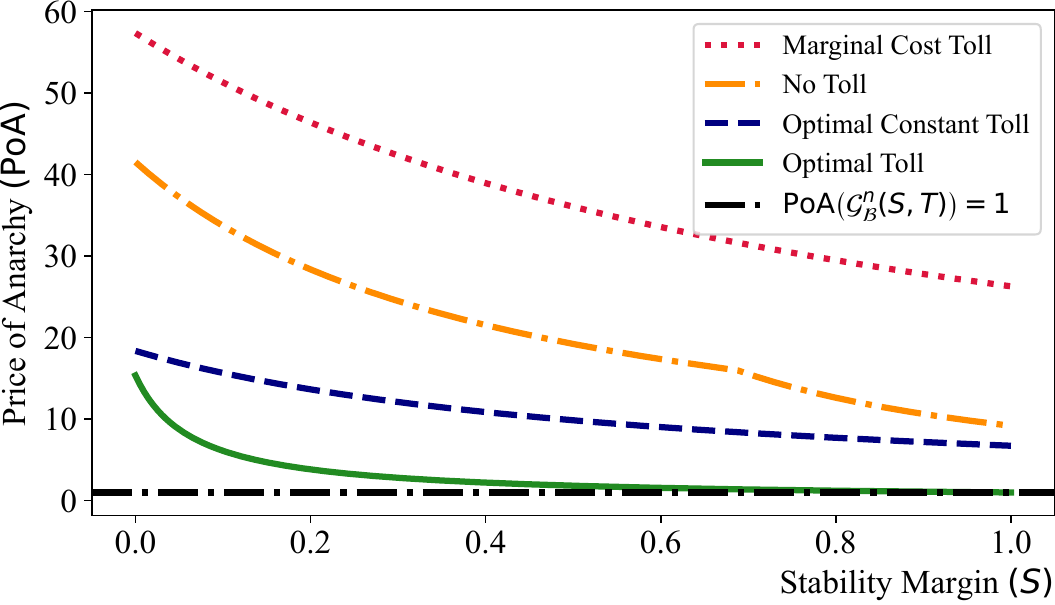}
        \end{center}
        \caption{Upper bound on the PoA from Theorem~\ref{thm:LP_Payoff_Stability} as a function of stability margin $S$ for various tolling mechanisms.
        We compare the optimal and marginal cost tolling mechanisms along with no tolls for the class of games $\mathcal{G}_{\mathcal{B}}^n(S, T)$, where \mbox{$\mathcal{B} = \{1, x, x^2, x^3\}$} and $n=8$.
        The bounds indicate that as the agents' aggregate satisfaction increases, the system cost at equilibrium becomes increasingly efficient.
        While $n$ is fixed here, the bounds can be applied to an arbitrary number of agents, as marginal cost tolls are independent of $n$ and optimal tolls are calculated by solving the LP for a fixed $n$ and are then extended analytically to any number of agents~\cite{paccagnan2021optimal}.
        }
        \label{fig:PoA_vs_Stability_Comparing_Tolls}
    \vspace{-4mm}
\end{figure}

\subsection{The PoA for Atomic Congestion Games without Tolls}
\label{ssec:PoA_NoTolls}

Fig.~\ref{fig:PoA_vs_Stability_No_Tolls} verifies Proposition~\ref{prop:PoA_S_LUB} by showing that for class of games $\mathcal{G}_\mathcal{B}^n(S, 0)$ with basis functions \mbox{$\mathcal{B} = \{1, \ldots, x^d\}$}, where \mbox{$d \in \{1, \ldots, 5\}$}, the bound produced by the LP in~\eqref{eq:LP_S} is strictly lower than the bound in~\eqref{eq:PoA_S_Seaton} for specific values of $S$.
When the LP in~\eqref{eq:LP_S} generates an upper bound on the PoA, it also implicitly generates game instances which achieve the bounds exactly, suggesting that the numerical bounds are tight.
We aim to address this claim in future work.

\begin{figure}[t]
    \vspace{2mm}
        \begin{center}
            \includegraphics[width=0.48\textwidth]{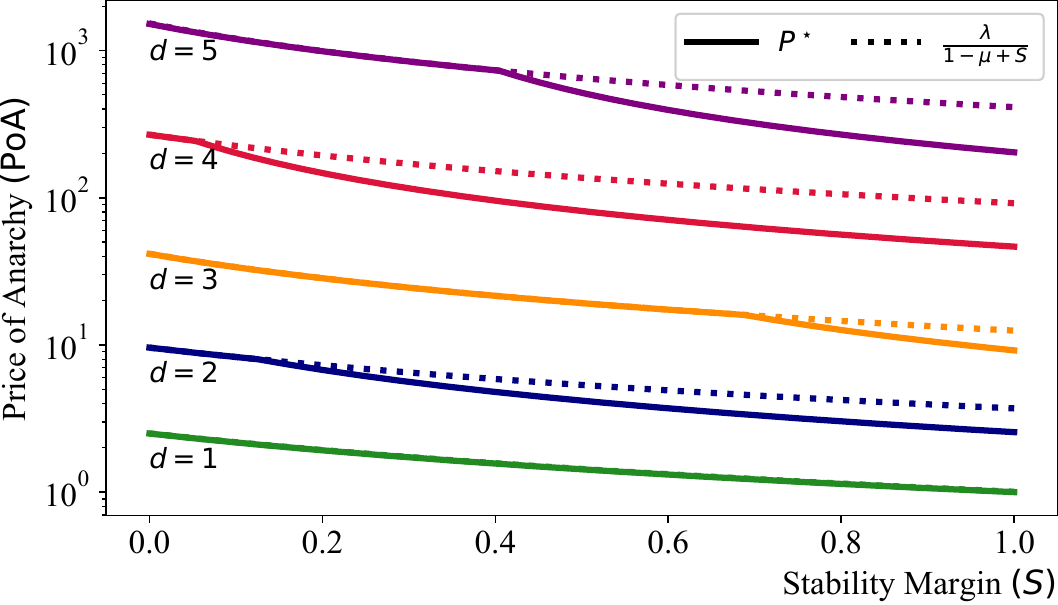}
        \end{center}
        \caption{Upper bound on the PoA as a function of stability margin $S$ when no toll is applied.
        For each class of atomic congestion games with polynomial basis functions of degree at most $d \in \{1, \ldots, 5\}$, the upper bound given by Theorem~\ref{thm:LP_Payoff_Stability} is compared to the previously best-known upper bound, given in equation~\eqref{eq:PoA_S_Seaton}.
        For each $d$, the plots demonstrate that $P^\star < \lambda / (1 - \mu + S)$ for certain values of $S$, proving Proposition~\ref{prop:PoA_S_LUB}.
        Here, the values of $\lambda$ and $\mu$ are generated using a grid search implementation of~\cite[Lemma 5.8]{aland2011exact}.
        }
        \label{fig:PoA_vs_Stability_No_Tolls}
    \vspace{-4mm}
\end{figure}

\section{Conclusion}
For incentivized atomic congestion games, we develop a linear program that computes an upper bound on the PoA as a function of the stability margin, which quantifies the aggregate satisfaction agents have with their equilibrium decisions.
We show that in games attaining the worst-case PoA, every agent is indifferent between their equilibrium action and their action in an optimal allocation; in particular, the stability margin is zero in these instances.
These results indicate that Nash equilibria achieving worst-case efficiency bounds are likely to be rare and may be avoided if agents randomize over cost-minimizing actions. Future work should determine whether a matching lower bound can be obtained so that the bound is tight, and extend the linear program to accommodate heterogeneous routing behavior.








\bibliographystyle{ieeetr}
\bibliography{library}

\end{document}